\documentclass[12pt, letterpaper]{article}
\usepackage{natbib}
\usepackage{amsmath,amssymb,amsfonts,amsthm,bbm}
\usepackage{epic,eepic,epsfig,longtable}
\usepackage{multirow,verbatim}
\usepackage{esvect}
\usepackage{array}
\usepackage{algorithm}
\usepackage{multirow}
\usepackage{footnote}
\usepackage{threeparttable} 
\usepackage{epsfig}
\usepackage{setspace}
\usepackage{color}
\usepackage{fancyhdr}
\usepackage{appendix}
\usepackage{listings}
\usepackage{longtable}
\usepackage{url}
\usepackage{lineno,hyperref}
\usepackage{subfigure}
\usepackage{graphicx}
\usepackage{booktabs}
\usepackage{longtable}
\usepackage{siunitx}
\usepackage{amsthm}
\usepackage{float}

\usepackage{epsfig}
\usepackage{setspace}
\usepackage{color}
\usepackage{fancyhdr}
\usepackage{appendix}
\usepackage{listings}
\usepackage{longtable}
\usepackage{url}
\usepackage{lineno,hyperref}
\usepackage{subfigure}
\usepackage{mathrsfs}
\usepackage{stmaryrd}
\usepackage{appendix}
\usepackage{algorithm}

\usepackage{enumitem}
\usepackage{algpseudocode}
\usepackage{algorithmicx}
\usepackage{caption}
\usepackage{amsmath}
\usepackage{graphics}
\usepackage{epsfig}
\usepackage{footnote}
\usepackage{threeparttable} 

\algdef{SE}[DOWHILE]{Do}{doWhile}{\algorithmicdo}[1]{\algorithmicwhile\ #1}%

\makeatletter
\def\singlespace{\def\baselinestretch{1}\@normalsize}

\makeatletter
\def\singlespace{\def\baselinestretch{1}\@normalsize}

\newcommand{\bfm}[1]{\ensuremath{\mathbf{#1}}}

\def\ba{\bfm a}   \def\bA{\bfm A}  
\def\bb{\bfm b}     
     
   \def\bD{\bfm D}  
     \def\EE{\mathbb{E}}

   \def\bH{\bfm H}  
     
   \def\bJ{\bfm J}  
     
   \def\bL{\bfm L}

   \def\bP{\bfm P}

   \def\bU{\bfm U}

\def\bx{\bfm x}   \def\bX{\bfm X}

\def\bzero{\bfm 0}

\newcommand{\bfsym}[1]{\ensuremath{\boldsymbol{#1}}}

\def\balpha{\bfsym \alpha}
\def\bbeta{\bfsym \beta}
\def\bgamma{\bfsym \gamma}

\def\btheta{\bfsym {\theta}}           
          
             \def\bSigma{\bfsym \Sigma}
\def\blambda {\bfsym {\lambda}}        
          \def\bOmega {\bfsym {\Omega}}

		        \def\bPsi{\bfsym {\Psi}}

\DeclareMathOperator{\Var}{Var}

\def\today{\ifcase\month\or
	January\or February\or March\or April\or May\or June\or
	July\or August\or September\or October\or November\or December\fi
	\space\number\day, \number\year}
\def\Cov{\mbox{Cov}}

\newdimen\biblioindent    \biblioindent=30pt

 at 8truept

\newcommand{\beq}{\begin{equation}\tag{1}}
\newcommand{\eeq}{\end{equation}}
\newcommand{\beqn}{\begin{eqnarray}}
\newcommand{\eeqn}{\end{eqnarray}}
\newcommand{\beqnn}{\begin{eqnarray*}}
\newcommand{\eeqnn}{\end{eqnarray*}}

\def\aVar{{\rm aVar}}
\def\bPsi{\boldsymbol{\Psi}}
\def\bX{\boldsymbol{X}}
\def\bx{\boldsymbol{x}}
\allowdisplaybreaks
\usepackage{verbatim}
\newcounter{CondCounter}
\def\Tsc{\mathcal{T}}
\def\Ssc{\mathcal{S}}
\def\trans{^{\scriptscriptstyle \sf T}}
\def\X{\boldsymbol{X}}
\def\x{\boldsymbol{x}}

\newtheorem{lemma}{Lemma}
\newtheorem{assumption}{Assumption}

\newtheorem{theorem}{Theorem}
\newtheorem{remark}{Remark}

\def\bbeta{\boldsymbol{\beta}}
\def\balpha{\boldsymbol{\alpha}}
\def\bgamma{\boldsymbol{\gamma}}
\def\btheta{\boldsymbol{\theta}}

\def\bphi{\boldsymbol{\phi}}

\def\bSigma{\boldsymbol{\Sigma}}

\def\bSigmahat{\widehat{\bSigma}}

\def\bPhi{\boldsymbol{\Phi}}
\def\bPsi{\boldsymbol{\Psi}}
\def\balphabar{\bar{\balpha}}

\def\bS{\boldsymbol{S}}

\def\muhat{\widehat{\mu}}

\def\subPS{_{\sf \scriptscriptstyle PS}}
\def\subOR{_{\sf \scriptscriptstyle OR}}
\def\subDR{_{\sf \scriptscriptstyle DR}}
\def\subPAD{_{\sf \scriptscriptstyle PAD}}
\def\subOAD{_{\sf \scriptscriptstyle OAD}}
\def\OAD{{\sf \scriptscriptstyle OAD}}
\def\subaug{_{\sf \scriptscriptstyle aug}}

\begin{document}

\title{Improve Efficiency of Doubly Robust Estimator when Propensity Score is Misspecified}
\author{Liangbo Lyu$^*$ \ and \ Molei Liu$^{\dag}$}

\footnotetext[1]{Liangbo Lv is an undergraduate student from the School of Statistics, Renmin University of China.}
\footnotetext[2]{Molei Liu is an assistant professor at Columbia University Mailman School of Public Health.}

\date{}

\maketitle
\begin{abstract}
\noindent Doubly robust (DR) estimation is a crucial technique in causal inference and missing data problems. We propose a novel $\bP$ropensity score $\bA$ugmented $\bD$oubly robust (PAD) estimator to enhance the commonly used DR estimator for average treatment effect on the treated (ATT), or equivalently, the mean of the outcome under covariate shift. Our proposed estimator attains a lower asymptotic variance than the conventional DR estimator when the propensity score (PS) model is misspecified and the outcome regression (OR) model is correct while maintaining the double robustness property that it is valid when either the PS or OR model is correct. These are realized by introducing some properly calibrated adjustment covariates to linearly augment the PS model and solving a restricted weighted least square (RWLS) problem to minimize the variance of the augmented estimator. Both the asymptotic analysis and simulation studies demonstrate that PAD can significantly reduce the estimation variance compared to the standard DR estimator when the PS model is wrong and the OR is correct, and maintain close performance to DR when the PS model is correct. We further applied our method to study the effects of eligibility for 401(k) plan on the improvement of net total financial assets using data from the Survey of Income and Program Participation of 1991. 
\end{abstract}


\baselineskip=17pt

\newpage

\section{Introduction}

\subsection{Background}\label{sec:intro:back}

Doubly robust (DR) estimation has attracted extensive interest in the literature on semiparametric theory and causal inference and is frequently used in biomedical science, economics, and policy science studies. It incorporates two nuisance models, a propensity score (PS) model, and an outcome regression (OR) model to characterize distributions of the exposure and outcome against the adjustment covariates respectively, and draws valid inferences when either one of them is correctly specified. It has been well-established that when both the PS and OR models are correct, the DR estimator is semiparametric efficient and its asymptotic variance does not really depend on the estimating equations for the nuisance models \citep[e.g.]{tsiatis2006semiparametric}. Nevertheless, there still remains an intriguing question on how to improve the asymptotic efficiency of the DR estimator when one nuisance model is misspecified. For the scenario with correct PS and wrong OR models, there is a track of work \citep[e.g.]{cao2009improving,tan2010bounded} proposing the so-called intrinsic efficient estimator that will be reviewed in Section \ref{sec:intro:lit}. This type of estimator preserves the double robustness property and achieves improved efficiency over the standard DR estimator when the PS model is correct and the OR is wrong. Interestingly, we notice that the dual problem of this, i.e., improving the (intrinsic) efficiency of the DR estimator under wrong PS and correct OR, is supposed to be equally important but has not been handled yet due to certain technical reasons that will be discussed later. Aimed in this paper, filling this methodological blank can effectively complement the existing tools for DR and semiparametric inference.

\subsection{Problem Setup}\label{sec:intro:set}

To make our idea easier to understand, we focus on a specific missing data problem: transfer estimation of the outcome's mean in the presence of covariate shift \citep[e.g.]{huang2007correcting}. This is also equivalent to estimating the average treatment effect on the treated (ATT) \citep[e.g.]{hahn2004functional} in the context of causal inference and matching-adjusted indirect comparison frequently conducted in biomedical studies \citep{signorovitch2010comparative}. Our method could be generalized to other settings such as estimating the average treatment effect (ATE) and transfer learning of regression models \citep{liu2020doubly}.

Suppose there are $n$ labeled samples with observed outcome $Y$ and covariates $\X\in\mathbb{R}^d$, and $N$ unlabeled samples only observed on $\X$. Let $\Delta=1$ indicate that the sample is labeled and $\Delta=0$ otherwise. The labeled observations $(Y_i,\X_i)$ are collected from a source population $\Ssc$ with $\Delta_i=1$ for $i=1,2,\ldots,n$. Assume $(Y_i,\X_i)\sim p_{\Ssc}(\x)q(y|\x)$ for $i=1,2,\ldots,n$ where $p_{\Ssc}(\x)$ and $q(y|\x)$ represent the density of $\X$ on $\Ssc$ and the conditional density of $Y$ given $\X=\x$ respectively. Meanwhile, there are unlabeled samples from a target population $\Tsc$ indicated by $\Delta_i=0$ and only observed on covariates $\X_i$ for $i=n+1,\ldots,N+n$. Assume that on $\Tsc$, $(Y_i,\X_i)\sim p_{\Tsc}(\x)q(y|\x)$ with $p_{\Tsc}(\x)$ representing the density of $\X$ on $\Tsc$ and the distribution of $Y\mid\X$ remaining to be the same as that on $\Ssc$. Our goal is to estimate $\mu_0=\EE_{\Tsc}Y$, the marginal mean of $Y$ on $\Tsc$. In the absence of observed $Y$ on the target samples, two simple strategies to estimate $\mu_0$ are introduced below. 
\begin{enumerate}
    \item[(PS)] Define the propensity score (PS) or density ratio between the two populations as $r_0(\x)={p_{\Tsc}(\x)}/{p_{\Ssc}(\x)}$. Estimate $r_0(\x)$ with some $\widehat r(\x)$ and average the observed $Y_i$ weighted by $\widehat r(\X_i)$ over $i=1,2,\ldots,n$ from $\Ssc$.
    
    \item[(OR)] Define the outcome regression (OR) or imputation model for $Y$ as $m_0(\x)=\mathbb{E}[Y\mid\X=\x]$. Estimate $m_0(\x)$ with some $\widehat m(\x)$ obtained using the labeled samples and average $\widehat m(\X_i)$ over $i=n+1,\ldots,n+N$ from $\Tsc$.
\end{enumerate}
Both the PS and OR strategies are built upon the assumption that the distribution of $Y\mid\X$ is the same between $\Ssc$ and $\Tsc$ so the knowledge of $Y$ on $\Ssc$ is transferable to $\Tsc$. This is in the same spirit as the \emph{no unmeasured confounding} assumption in the context of causal inference.

\subsection{Related literature}\label{sec:intro:lit}

Our work is based on the doubly robust (DR) inference framework that has been frequently studied and applied in the past years \citep[e.g.]{robins1994estimation,bang2005doubly,kang2007demystifying,tan2010bounded,vermeulen2015bias}. It combines the PS and OR models introduced in Section \ref{sec:intro:set} to construct an estimator that is valid when at least one of the two nuisance models are correct and, thus, regarded as a more robust statistical inference procedure than the simple PS and OR strategies. Early work in DR inference \citep[e.g.]{bang2005doubly,kang2007demystifying} mainly used \emph{working} low-dimensional parametric regression to construct the PS and OR models. Recent progress has been made to accommodate the use of high-dimensional regression or complex machine learning methods in estimating the nuisance models \citep[e.g.]{chernozhukov2018double,tan2020model}, which is less prone to model misspecification. We focus the scope of this paper on the low-dimensional parametric setting that is technically less involved but more user-friendly and less sensitive to over-fitting in practice. It is also possible and valuable to generalize our work to the settings of high-dimensional parametric \citep[e.g.]{tan2020model,dukes2020inference} or semi-non-parametric \citep{liu2020doubly} nuisance models, in which model misspecification is still an important concern.

There has risen great interest in studying and improving the asymptotic efficiency of the DR estimator. One track of literature studied the local efficiency of the DR estimator, i.e., if it is semiparametric efficient when both the PS and OR models are known or correctly specified. While it was shown that the standard DR estimator for the ATE \citep{robins1994estimation} achieves such local efficiency \citep{hahn1998role,tsiatis2006semiparametric}. This result cannot be directly applied to the ATT estimator because unlike ATE, the PS model of ATT is informative (or non-ancillary) \citep{hahn1998role,hahn2004functional}. \cite{shu2018improved} further studied this subtle issue and proposed locally efficient DR estimators for ATT based on its influence function. 

Meanwhile, another track of literature focuses on improving the efficiency of the DR estimator in the presence of correct PS and potentially wrong OR models and, thus, is more relevant to our work that also aims at automatic variance reduction under model misspecification. A class of \emph{intrinsic efficient} DR estimator has been proposed for the efficient estimation of ATE \citep{cao2009improving,tan2010bounded}, ATT \citep{shu2018improved}, casual regression model \citep{rotnitzky2012improved}, longitudinal data \citep{han2016intrinsic}, individual treatment rule \citep{pan2021improved}, etc. This type of estimator is (i) valid when either nuisance model is correct; (ii) equivalent with the standard DR estimator when both models are correct; and (iii) of the minimum variance under correct PS and wrong OR, among all the DR estimators with the same parametric specification of the OR model, and, consequently, more efficient than the standard DR estimator. In addition, it was shown that including more prognostic covariates or auxiliary basis in the PS model can always help to reduce the variance of the ATE estimator \citep{hahn2004functional,tsiatis2006semiparametric}. Motivated by this, \cite{cheng2020estimating} proposed a double-index PS estimator for ATE that smooths the treatment over the parametric PS and OR models to achieve the DR property as well as variance reduction under correct PS and wrong OR. Nevertheless, such a strategy may also incur over-fitting issues and cause poor performance in finite or small sample studies \citep{gronsbell2022efficient}.

Although the correct PS and wrong OR setting has been frequently studied, there is still a paucity of solutions to its dual problem, i.e., enhancing the DR estimator under the wrong PS and correct OR. Some early work like \cite{kang2007demystifying} and \cite{cao2009improving} argued that the simple OR strategy is an ideal choice when one knows the PS model is wrong since it is free of PS weighting that may decrease the effective sample size. However, since there are no perfect ways to examine model correctness without any additional assumptions, this strategy can never be as robust as the DR estimator to misspecification of the OR model.

We also notice a large body of work in statistical learning and causal inference that aims at leveraging some auxiliary data or information to boost the asymptotic efficiency of certain estimators using the idea of augmentation. For example, \cite{kawakita2013semi}, \cite{chakrabortty2018efficient} and \cite{azriel2021semi} proposed different semi-supervised learning methods that improve estimation efficiency of the linear model leveraging large unlabeled data drawn from the same distribution as the labeled samples. Methods like \cite{chen2000unified} and \cite{yang2019combining} utilized external data with error-prone outcomes or covariates to construct control variate for variance reduction. These methods, as well as other examples, rely on some auxiliary data to construct estimators that always converge to zero and are asymptotically correlated with the target estimator. These zero estimators are then used to augment the target estimator properly for variance reduction. Our work also adapts the high-level idea of augmentation. But different from these methods, ours does not leverage any auxiliary samples or knowledge and additionally cares about the need of prioritizing validity (double robustness) over statistical power. Consequently, the asymptotic behavior of our augmented estimator actually varies according to the correctness of the nuisance models and is more technically involved in to study.

\subsection{Our contribution}

To estimate $\mu_0$ introduced in Section \ref{sec:intro:set} efficiently, we propose a novel $\bP$ropensity score $\bA$ugmented $\bD$oubly robust (PAD) estimation method that enhances the standard DR estimator of $\mu_0$ by linearly augmenting the PS model with some functions of $\X$. Both the augmentation functions and their linear coefficients are wisely and carefully constructed such that the augmentation term always reduces the variance of the DR estimator if the PS is wrong and the OR is correct while it automatically converges to zero if the PS is correct, in order to avoid bias and ensure double robustness. Also, when both models are correct, our PAD estimator becomes asymptotically equivalent to the standard DR estimator. To our best knowledge, the proposed estimator is the first one to simultaneously have the DR property and a smaller variance than the standard DR estimator under wrong PS and correct OR models. Thus, our work serves as an important complement to existing DR inference approaches, especially to the intrinsically efficient DR estimators proposed to work for the setting with correct PS and wrong OR \citep[e.g.]{cao2009improving,tan2010bounded}.

\section{Method}\label{sec:method}

\subsection{Doubly robust estimator}

As a prerequisite of our proposal, we first introduce the standard DR estimator for $\mu_0$ under the setup described in Section \ref{sec:intro:set}, which has been studied for years \citep[e.g.]{hahn1998role,hahn2004functional,shu2018improved}. Following a common strategy \citep[e.g.]{bang2005doubly,shu2018improved,liu2020doubly}, we form the PS and OR models as $r(\x)=\exp(\x\trans\bgamma)$ and $m(\x)=g(\bx\trans\balpha)$ where $\bgamma$ and $\balpha$ are model coefficients and $g(\cdot)$ is a known and differentiable link function. We say that the PS (or OR) model is correct if there exists $\bgamma_0$ (or $\balpha_0$) such that the true $r_0(\x)=\exp(\x\trans\bgamma_0)$ (or $m_0(\x)=g(\bx\trans\balpha_0)$). Denote the empirical mean operator on $\Ssc$ and $\Tsc$ as $\widehat\EE_{\Ssc}$ and $\widehat\EE_{\Tsc}$ such that 
\[
\widehat\EE_{\Ssc}a(\X,Y)=n^{-1}\sum_{i=1}^na(\bX_i,Y_i),\quad\widehat\EE_{\Tsc}a(\X,Y)=N^{-1}\sum_{i=n+1}^{n+N}a(\bX_i,Y_i)
\]
for any function $a(\cdot)$. Suppose the two nuisance estimators $\widehat\bgamma$ and $\widehat\balpha$ are obtained respectively by solving the estimating equations:
\begin{equation}
\widehat\EE_{\Ssc}\X\exp(\X\trans\bgamma)=\widehat\EE_{\Tsc}\X,\quad\widehat\EE_{\Ssc}\X\{Y-g(\X\trans\balpha)\}=\bzero.
\label{equ:2.1.1}
\end{equation}
The estimating equations for $\bgamma$ in \eqref{equ:2.1.1} is usually referred as covariate balancing \citep{imai2014covariate,zhao2017entropy}, and those for $\balpha$ correspond to the ordinary least square regression when $g(a)=a$ and the logistic regression when $Y$ is binary and $g(a)={\rm expit}(a)=e^a/(1+e^a)$. Note that one can use alternative estimation procedures to obtain $\bgamma$ and $\balpha$, e.g., running a logistic regression on $\Delta$ against $\X$ to estimate $\bgamma$, and our proposed method could naturally adapt to different choices on this.

Based on $\widehat\bgamma$ and $\widehat\balpha$, the PS and OR estimators introduced in Section \ref{sec:intro:set} can be specified as $\muhat\subPS=\widehat\EE_{\Ssc}Y\exp(\X\trans\widehat\bgamma)$ and $\muhat\subOR=\widehat\EE_{\Tsc}g(\X\trans\widehat\balpha)$ respectively. Then the standard DR estimator is constructed by augmenting one of them with another nuisance model:
\begin{equation}
\muhat\subDR=\widehat\EE_{\Ssc}\{Y-g(\X\trans\widehat\balpha)\}\exp(\X\trans\widehat\bgamma)+\widehat\EE_{\Tsc}g(\X\trans\widehat\balpha).
\label{equ:2.1.3}
\end{equation}
When the PS model is correct and $\widehat\bgamma$ converges to $\bgamma_0$,  $\widehat\EE_{\Tsc}g(\X\trans\widehat\balpha)-\widehat\EE_{\Ssc}g(\X\trans\widehat\balpha)\exp(\X\trans\widehat\bgamma)$ converges to zero and the remainder term $\widehat\EE_{\Ssc}Y\exp(\X\trans\widehat\bgamma)$ is exactly the PS estimator converging to $\mu_0$. Similarly, when OR is correct, we can show that $\widehat\EE_{\Ssc}\{Y-g(\X\trans\widehat\balpha)\}\exp(\X\trans\widehat\bgamma)$ converges to zero and $\widehat\EE_{\Tsc}g(\X\trans\widehat\balpha)$ converges to $\mu_0$. Thus $\muhat\subDR$ is doubly robust in the sense that it is consistent when either the PS or OR model is correctly and consistently estimated.

\subsection{Expansion of DR estimator under correct OR model}\label{sec:method:expand}

To help the readers understand our method more intuitively, we now heuristically derive and analyze the asymptotic expansion of $\muhat\subDR$ when the OR model is correctly specified. Suppose that $\widehat\bgamma$ and $\widehat\balpha$ converge to some $\bar\bgamma$ and $\bar\balpha$ defined as the solutions to the population-level estimating equations $\EE_{\Ssc}\X\exp(\X\trans\bgamma)=\EE_{\Tsc}\X$ and $\EE_{\Ssc}\X\{Y-g(\X\trans\balpha)\}=\bzero$, respectively. Let $\widehat r(\x)=\exp(\X\trans\widehat\bgamma)$, $\bar r(\x)=\exp(\X\trans\bar\bgamma)$, and $\bS(\balpha)=\bS(Y,\X,\balpha)=\X\{Y-g(\X\trans\balpha)\}$. Suppose that the OR model is correct, i.e., $m_0(\x)=g(\X\trans\balpha_0)$ and $\balpha_0=\bar\balpha$, and $n^{1/2}(\widehat\balpha-\bar\balpha,\widehat\bgamma-\bar\bgamma)$ is asymptotically normal with mean zero following the standard M-estimation theory \citep{van2000asymptotic}. Then we have
\[
\widehat\EE_{\Ssc}\{Y-g(\X\trans\widehat\balpha)\}\{\widehat r(\X)-\bar r(\X)\}=o_p(n^{-1/2})
\]
due to Neyman orthogonality \citep{neyman1959optimal}, which, as will be strictly proved in Section \ref{sec:thm}, implies that $\muhat\subDR$ defined in (\ref{equ:2.1.3}) is asymptotically equivalent with 
\begin{align*}
   \widetilde\mu\subDR=&\widehat\EE_{\Ssc}\{Y-g(\X\trans\bar\balpha)\}\bar r(\X)+\widehat\EE_{\Tsc}g(\X\trans\bar\balpha)\\
   &+\left[\widehat\EE_{\Ssc}\{g(\X\trans\bar\balpha)-g(\X\trans\widehat\balpha)\}\bar r(\X)+\widehat\EE_{\Tsc}\{g(\X\trans\widehat\balpha)-g(\X\trans\bar\balpha)\}\right]\\
   \approx&\widehat\EE_{\Ssc}\{Y-g(\X\trans\bar\balpha)\}\bar r(\X)+\widehat\EE_{\Tsc}g(\X\trans\bar\balpha)+\bL\trans\widehat\EE_{\Ssc}\X\{Y-g(\X\trans\bar\balpha)\},
\end{align*}
where $\bL=-\bar\bH^{-1}\left\{\EE_{\Ssc}\X \dot g(\X\trans\balphabar)\bar r(\X)-\EE_{\Tsc}\X \dot g(\X\trans\balphabar)\right\}$, $\bar\bH=\EE_{\Ssc}\X\X\trans\dot g(\X\trans\balphabar)$, and $\dot g(a)$ is the derivative of $g(a)$. To derive the above result, we use the standard asymptotic expansion of $\widehat\balpha$ given by our Lemma \ref{lem:b3} in Appendix, and the symbol ``$\approx$" indicates that the difference between the two lines is up to $o_p(n^{-1/2})$ and, thus, asymptotically negligible. So when OR is correct, the asymptotic variance of $n^{1/2}(\widehat\mu\subDR-\mu_0)$ is equal to that of $n^{1/2}(\widetilde\mu\subDR-\mu_0)$, which can be expressed as
\begin{equation}
 \aVar\{n^{1/2}(\widehat\mu\subDR-\mu_0)\}=\EE_{\Ssc}\{\bar r(\X)\}^2 v(\X)+2\bL\trans \EE_{\Ssc}\X\bar r(\X)v(\X)+C,
\label{equ:2.1.4}
\end{equation}
where $v(\x)=\Var(Y\mid\X)$ and $C$ is some positive constant free of $\bar r(\cdot)$ and, thus, needs not to be considered in the following derivation. Note that when the PS model also is correct, i.e., $\bar r(\cdot)=r_0(\cdot)$, we further have $\bL=\bzero$. 

Empirically, term $\bL$ in (\ref{equ:2.1.4}) can be estimated by 
\begin{equation}
\widehat\bL=-\widehat\bH^{-1}\left\{\widehat\EE_{\Ssc}\X \dot g(\X\trans\widehat\balpha)\exp(\X\trans\widehat\bgamma)-\widehat\EE_{\Tsc}\X \dot g(\X\trans\widehat\balpha)\right\},
\label{equ:L}
\end{equation}
where $\widehat\bH=\widehat\EE_{\Ssc}\X\X\trans\dot g(\X\trans\widehat\balpha)$. Estimation of $v(\x)$ relies on our \emph{working} assumption on the form of $\Var(Y\mid\X)$. For example, one may assume $Y=m_0(\X)+\epsilon$ where $\epsilon\sim {\rm N}(0,\sigma^2)$ so $v(\x)$ is invariant of $\x$ and can be simply imputed with the moment estimator of $\sigma^2$. Also, for the common Poisson model $Y\sim {\rm Poisson}\{\exp(\X\trans\balpha_0)\}$ and logistic model $Y\sim {\rm Bernoulli}\{{\rm expit}(\X\trans\balpha_0)\}$, one can naturally estimate $v(\x)$ by $\exp(\x\trans\widehat\balpha)$ and ${\rm expit}(\x\trans\widehat\balpha)\{1-{\rm expit}(\x\trans\widehat\balpha)\}$ respectively. To preserve generality, we introduce a \emph{working} model $v_{\btheta}(\x)$ for $v(\x)$ with some nuisance parameter $\btheta$ to be estimated as $\widehat\btheta$ that could be partially or fully determined by $\widehat\balpha$. Suppose that $\widehat\btheta$ converges to some $\bar\btheta$. As will be shown in Section \ref{sec:thm}, violation of this conditional variance model, i.e., $v(\x)\neq v_{\bar\btheta}(\x)$ does not impact the double robustness of our proposed estimator but only affects its efficiency gain when PS is wrong and OR is correct.

\subsection{PAD estimator}\label{sec:method:pad}

Now we formally introduce the propensity score augmented doubly robust (PAD) estimator. Our central idea is to augment the PS model $\bar r(\X)=\exp(\X\trans\bar\bgamma)$ as $\bar r\subaug(\X;\bbeta)=\exp(\X\trans\bar\bgamma)+\bPsi\trans\bbeta$ and use $\bar r\subaug(\cdot)$ to replace $\bar r(\cdot)$ in the DR estimator. Here $\bPsi$ is some properly constructed basis function of $\X$ and $\bbeta$ is some loading coefficient vector to be estimated. We first describe the empirical construction procedures for PAD in Algorithm \ref{alg:1} and then discuss the reason and intuition of the key steps in this algorithm. 

\begin{algorithm}[htb!]
\caption{\label{alg:1} Propensity score Augmented Doubly robust (PAD) estimation}
\begin{algorithmic}
 \State[Step 1] Solve the estimating equations in (\ref{equ:2.1.1}) to obtain $\widehat\bgamma$ and $\widehat\balpha$, and obtain the conditional variance estimator as $\widehat\btheta$.
~\\
 \State[Step 2] Specify $\bPhi=\bphi(\X)$ of larger dimensionality than $\X$ using any basis function $\bphi(\cdot)$, and take $\widehat\bPsi=\bPhi-{\widehat\EE_{\Tsc}[\bPhi v_{\widehat\btheta}(\X)]}/{\widehat\EE_{\Tsc}v_{\widehat\btheta}(\X)}$.
 ~\\
 \State[Step 3] Solve the restricted weighted least square (RWLS) problem:
 \begin{equation}
 \widehat\bbeta={\rm argmin}_{\bbeta} \widehat V_{\mu}(\bbeta),\quad\mbox{s.t.}\quad\widehat\EE_{\Ssc}\X \dot{g}(\X\trans\widehat\balpha)\widehat\bPsi\trans\bbeta=\bzero,
 \label{equ:rwls}
 \end{equation}
where
\begin{equation}
\widehat V_{\mu}(\bbeta)=\widehat\EE_{\Ssc}\{\exp(\X\trans\widehat\bgamma)+\widehat\bPsi\trans\bbeta\}^2 v_{\widehat\btheta}(\X)+2\widehat\bL\trans \widehat\EE_{\Ssc}\X\{\exp(\X\trans\widehat\bgamma)+\widehat\bPsi\trans\bbeta\}v_{\widehat\btheta}(\X),
\label{equ:vbeta}
\end{equation}
and $\widehat\bL$ is as defined in equation (\ref{equ:L}).
 ~\\
 \State[Step 4] Obtain the PAD estimator through
\[
 \widehat\mu\subPAD=\widehat\EE_{\Ssc}\{Y-g(\X\trans\widehat\balpha)\}\{\exp(\X\trans\widehat\bgamma)+\widehat\bPsi\trans\widehat\bbeta\}+\widehat\EE_{\Tsc}g(\X\trans\widehat\balpha).
 \] 
\end{algorithmic}
\end{algorithm}
For heuristic analysis, suppose that all estimators used in (\ref{equ:vbeta}) converge to their limiting values. Then let $\bPsi=\bPhi-{\EE_{\Tsc}[\bPhi v_{\bar\btheta}(\X)]}/{\EE_{\Tsc}v_{\btheta}(\X)}$ be the limits of $\widehat\bPsi$, $\bar\bbeta$ the limits of $\widehat\bbeta$, with its specific form given by Lemma \ref{lem:b1} in Appendix, and
\[
 V_{\mu}(\bbeta)=\EE_{\Ssc}\{\exp(\X\trans\bar\bgamma)+\bPsi\trans\bbeta\}^2 v_{\bar \btheta}(\X)+2\bL\trans \EE_{\Ssc}\X\{\exp(\X\trans\bar\bgamma)+\bPsi\trans\bbeta\}v_{\bar \btheta}(\X)
\]
the limiting function of $\widehat V_{\mu}(\bbeta)$ specified in Algorithm \ref{alg:1}. We shall consider two scenarios separately to demonstrate that our proposed PAD estimator not only maintains double robustness property but also has a lower asymptotic variance than $\widehat\mu\subDR$ when the OR model is correctly specified and PS is wrong. Rigorous justification for these results will be provided in Section \ref{sec:thm}.

\paragraph{Correct PS model.} When the PS model is correct, we easily have $\bL=\bzero$ as stated in Section \ref{sec:method:expand} so $V_{\mu}(\bbeta)=\EE_{\Ssc}\{\exp(\X\trans\bar\bgamma)+\bPsi\trans\bbeta\}^2 v_{\bar \btheta}(\X)$, and 
\[
\frac{\partial V_{\mu}(\bbeta)}{\partial\bbeta}=\EE_{\Ssc}\bPsi\exp(\X\trans\bar\bgamma)v_{\bar \btheta}(\X)=\EE_{\Tsc}\bPsi v_{\bar \btheta}(\X).
\]
By definition of $\bPsi$, we have $\EE_{\Tsc}\bPsi v_{\bar \btheta}(\X)=\bzero$, as ensured by the mean shift of $\bPhi$ in Step 2 of Algorithm \ref{alg:1}. Thus, $\bbeta=\bzero$ minimizes $V_{\mu}(\bbeta)$ and consequently, is the solution of the population-level version of the RWLS problem (\ref{equ:rwls}) since the linear constraints in (\ref{equ:rwls}) is trivially satisfied by $\bbeta=\bzero$. This implies that as long as the PS model is correct, $\widehat\bbeta$ converges to $\bzero$ so the augmented PS estimator $\exp(\X\trans\widehat\bgamma)+\widehat\bPsi\trans\widehat\bbeta$ converges to the correct PS model, which ensures $\widehat\mu\subPAD$ to converge to the true $\mu_0$. Meanwhile, it is clear that the augmentation of PS does not change the OR model at all. Therefore, $\widehat\mu\subPAD$ preserves the same DR property as $\widehat\mu\subDR$, i.e., being (root-$n$) consistent whenever the PS or the OR model is correctly specified.

\paragraph{Correct OR and wrong PS.} Note that $\widehat\mu\subPAD=\widehat\mu\subDR+\widehat\EE_{\Ssc}\widehat\bPsi\trans\widehat\bbeta\{Y-g(\X\trans\widehat\balpha)\}$ and when the OR model is correct,
\begin{equation}
\begin{split}
\widehat\EE_{\Ssc}\widehat\bPsi\trans\widehat\bbeta\{Y-g(\X\trans\widehat\balpha)\}=&\widehat\EE_{\Ssc}\widehat\bPsi\trans\widehat\bbeta\{Y-g(\X\trans\balpha_0)\}+\widehat\EE_{\Ssc}\widehat\bPsi\trans\widehat\bbeta\{g(\X\trans\balpha_0)-g(\X\trans\widehat\balpha)\}\\
\approx&\widehat\EE_{\Ssc}\bPsi\trans\bar\bbeta\{Y-g(\X\trans\balpha_0)\}+\widehat\EE_{\Ssc} (\balpha_0-\widehat\balpha)\trans\X \dot{g}(\X\trans\widehat\balpha)\widehat\bPsi\trans\widehat\bbeta,
\end{split}
\label{equ:decom}
\end{equation}
in which we use the orthogonality between $\widehat\bPsi\trans\widehat\bbeta-\bPsi\trans\bar\bbeta$ and $Y-g(\X\trans\balpha_0)$ on the first term, as well as expansion on $g(\X\trans\balpha_0)-g(\X\trans\widehat\balpha)$ in the second term of the first line, to derive the ``$\approx$" relation shown in the second line. Here, ``$\approx$" in (\ref{equ:decom}) again means that the difference between the first and second line is up to $o_p(n^{-1/2})$ and, thus, becomes asymptotically negligible. In addition, according to the moment constraint in the RWLS problem (\ref{equ:rwls}), $\widehat\EE_{\Ssc}\X \dot{g}(\X\trans\widehat\balpha)\widehat\bPsi\trans\widehat\bbeta$ converges to $\bzero$. So the second term in the second line of (\ref{equ:decom}) is also negligible and $\widehat\mu\subPAD\approx \widehat\mu\subDR+\widehat\EE_{\Ssc}\bPsi\trans\bar\bbeta\{Y-g(\X\trans\balpha_0)\}$. Combining this with equation (\ref{equ:2.1.4}) as well as the asymptotic equivalence between $\widehat\mu\subDR$ and $\widetilde\mu\subDR$ discussed in Section \ref{sec:method:expand}, we have
\begin{equation}
\aVar\{n^{1/2}(\widehat\mu\subPAD-\mu_0)\}=\EE_{\Ssc}\{\bar r(\X)+\bPsi\trans\bar\bbeta\}^2 v(\X)+2\bL\trans \EE_{\Ssc}\X\{\bar r(\X)+\bPsi\trans\bar\bbeta\}v(\X)+C,
\label{equ:var:pad}
\end{equation}
which, after dropping the invariant $C$, is equal to $V_{\mu}(\bar\bbeta)$, the limiting value of the minimized objective function $\widehat V_{\mu}(\widehat\bbeta)$ in the RWLS problem (\ref{equ:rwls}). Note that $\bbeta=\bzero$ is always feasible to the linear constraint in (\ref{equ:rwls}) and if we simply replace $\bar\bbeta$ with $\bzero$ in the right-hand side of (\ref{equ:var:pad}), it reduces to the asymptotic variance of $n^{1/2}(\widehat\mu\subDR-\mu_0)$ derived in (\ref{equ:2.1.4}). Meanwhile, when the PS model is wrong, ${\partial V_{\mu}(\bbeta)}/{\partial\bbeta}$ is typically not $\bzero$ at $\bbeta=\bzero$ so the population-level minimizer $\bar\bbeta\neq 0$. Thus, $\aVar\{n^{1/2}(\widehat\mu\subPAD-\mu_0)\}\leq \aVar\{n^{1/2}(\widehat\mu\subDR-\mu_0)\}$ when the OR model is correct and the strict ``$<$" will hold in general when the PS model is wrong.

\section{Asymptotic analysis}\label{sec:thm}

In this section, we rigorously present the asymptotic properties of the proposed PAD estimator and compare PAD with the standard DR estimator. We first introduce some mild and common regularity assumptions. Without loss of generality, we assume that $n/N=O(1)$ so the desirable parametric rate of the DR estimators will be $O(n^{-1/2})$.
\begin{assumption}
The supports of $\X$ and $\bPhi$ are compact and $\mathbb{E}Y^4<\infty$.
\label{asu:1}
\end{assumption}

\begin{assumption}
The link function $g(\cdot)$ is differentiable with derivative $\dot g(\cdot)$ and there exists a constant $L$ such that $|\dot g(x_1)-\dot g(x_2)|<L|x_1-x_2|$ for all $x_1,x_2\in\mathbb{R}$.
\label{asu:2}
\end{assumption}

\begin{assumption}
The dimension of $\bPsi$ is larger than that of $\X$. Matrices $\EE_{\Ssc}\{\bPsi\bPsi\trans v_{\bar \btheta}(\X)\}$, $\EE_{\Ssc}\{\X\X\trans\exp(\X\trans\bar\bgamma)\}$, $\EE_{\Ssc}\{\X\X\trans\dot g(\X\trans\bar\balpha)\}$ and $\EE_{\Ssc}\{\bPsi\X\trans\dot g(\X\trans\bar\balpha)\}$ have all their eigenvalues bounded and staying away from zero.
\label{asu:3}
\end{assumption}

\begin{assumption}
The conditional variance function $v_{\btheta}(\x)$ is differentiable on $\btheta$ with a bounded partial derivative $\partial_{\btheta} v_{\btheta}(\x)$. The estimator $\widehat\btheta$ converges to some $\bar\btheta$ in probability and satisfies that $n^{1/2}(\widehat\btheta-\bar\btheta)$ is asymptotic normal with mean zero.
\label{asu:4}
\end{assumption}

\begin{remark}
Assumptions \ref{asu:1}--\ref{asu:3} are all mild, standard, and commonly used to justify the asymptotic properties of M-estimation \citep{van2000asymptotic}. Note that in Assumption \ref{asu:3}, we take $\bPsi$ to have larger dimension than $\X$ and make regularity conditions on $\EE_{\Ssc}\{\X\X\trans\dot g(\X\trans\bar\balpha)\}$ and $\EE_{\Ssc}\{\bPsi\X\trans\dot g(\X\trans\bar\balpha)\}$. These are to ensure that $\widehat\bbeta$ is not zero and properly converges to $\bar\bbeta$. Assumption \ref{asu:4} constrains the way of specifying $v_{\btheta}(\x)$ and estimating $\btheta$. Under Assumptions \ref{asu:1}--\ref{asu:3}, this assumption is satisfied when either $\btheta$ is fully determined by $\balpha$, e.g., in a Poisson or logistic model for $Y$ against $\X$, or when $\btheta$ is estimated by additionally fitting some parametric model of $\Var(Y\mid\X)$ against $\X$. 
\end{remark}

Now we present the main results about the robustness and efficiency of our proposed PAD estimator in Theorem \ref{thm:1} with its proof given in Section \ref{sec:proof} of the Appendix. Some important heuristics of this theorem has already been discussed in Section \ref{sec:method:pad}.
\begin{theorem}
Under Assumptions \ref{asu:1}--\ref{asu:4}, it holds that
\begin{enumerate}
    \item[(i)] {\bf Double robustness}. When either the PS or the OR model is correctly specified, i.e., $r_0(\x)=\exp(\x\trans\bgamma_0)$ for some $\bgamma_0$ or $m_0(\x)=g(\x\trans\balpha_0)$ for some $\balpha_0$, $\muhat\subPAD\xrightarrow{p}\mu_0$ and $n^{1/2}(\muhat\subPAD-\mu_0)$ weakly converges to some normal distribution with mean zero.
    
    \item[(ii)] {\bf Variance reduction under wrong PS}. When the OR model is correct while the PS model may be misspecified, the asymptotic variance of $n^{1/2}(\muhat\subPAD-\mu_0)$ is always not larger than that of $n^{1/2}(\muhat\subDR-\mu_0)$. Further when $\bar\bbeta\neq 0$ (the explicit form of $\bar\bbeta$ is given in Lemma \ref{lem:b1}), $n^{1/2}(\muhat\subPAD-\mu_0)$ has a strictly smaller asymptotic variance than $n^{1/2}(\muhat\subDR-\mu_0)$.
    
    \item[(iii)] {\bf Equivalence under correct PS and OR}. When both the PS and OR models are correct, $n^{1/2}(\muhat\subPAD-\mu_0)$ and $n^{1/2}(\muhat\subDR-\mu_0)$ are asymptotically equivalent and have the same asymptotic variance.

\end{enumerate}
\label{thm:1}
\end{theorem}

\section{Simulation study}\label{sec:simu}
We conducted simulation studies to evaluate our proposed estimator and compare it with the standard DR estimator. In our studies, we generate covariates $\X=(X_1,X_2,X_3)\trans$ from ${\rm N}(\bzero,\bSigma)$ with $\bSigma=(\sigma_{ij})\in \mathbb{R}^{3\times 3}$ and $\sigma_{ij}=0.3^{|i-j|}$. For generation of the population assignment $\Delta$ and outcome $Y$, we consider six settings, namely:
\begin{enumerate}
    \item[(G1)] {\bf Gassuian $Y$, Correct PS, Correct OR.} $\text{Pr}(\Delta=1\mid \X)\}={\rm expit}(X_1-2X_2+X_3)$ and $Y=0.5X_1+0.5X_2+X_3+\epsilon$ where $\epsilon\mid\X\sim {\rm N}(0,1)$.

    \item[(G2)] {\bf Gassuian $Y$, Correct PS, Wrong OR.} $\text{Pr}(\Delta=1\mid \X)={\rm expit}(X_1-2X_2+X_3)$ and $Y=0.5X_1+0.5X_2+\sin(X_2+0.5X_3)+\epsilon$.

    \item[(G3)] {\bf Gassuian $Y$, Wrong PS, Correct OR.} $\text{Pr}(\Delta=1\mid \X)={\rm expit}(4+X_1+X_2+X_3-1.5|X_1|-1.5|X_2|-|X_3|)$ and $Y=0.5X_1+0.5X_2+X_3+\epsilon$.

    \item[(L1)] {\bf Binary $Y$, Correct PS, Correct OR.} $\text{Pr}(\Delta=1\mid \X)={\rm expit}(X_1-2X_2+X_3)$ and $\text{Pr}(Y=1\mid \X)={\rm expit}(0.5X_1+0.5X_2+X_3)$.

    \item[(L2)] {\bf Binary $Y$, Correct PS, Wrong OR.} $\text{Pr}(\Delta=1\mid \X)={\rm expit}(X_1-2X_2+X_3)$ and $\text{Pr}(Y=1\mid X)\}={\rm expit}(0.5X_1+0.5X_2+\sin(X_2+0.5X_3))$

    \item[(L3)] {\bf Binary $Y$, Wrong PS, Correct OR.} $\text{Pr}(\Delta=1\mid \X)={\rm expit}(4+X_1 + X_2 + X_3 - 1.5|X_1| - 1.5|X_2| -|X_3|)$ and $\text{Pr}(Y=1\mid \X)={\rm expit}(0.5X_1+0.5X_2+X_3)$.
    
\end{enumerate}
In Settings (G1)--(G3), $Y$ is a gaussian variable and we fit linear models for $Y\sim\X$ with $v_{\btheta}(\x)=1$. While in Settings (L1)--(L3), we fit logistic models for the binary $Y$ against $\X$ with $v_{\btheta}(\x)={\rm expit}(\X\trans\balpha)\{1-{\rm expit}(\X\trans\balpha)\}$. We consider different scenarios about the correctness of the PS and OR models to examine the robustness and efficiency of PAD. Bootstrap is used for estimating the asymptotic variance and constructing the confidence interval (CI). For effective variance reduction on PAD when PS is wrong, i.e. under Settings (G3) and (L3), we include in the augmentation covariates $\bPhi$ a decent amount of $\X$'s basis functions including $X_j$, $\exp(X_j)$, $|X_j|$, $\exp(-X_{j_1}-X_{j_2})$, and $\exp(-X_1-X_2-X_3)$ for all $j$ and $j_1\neq j_2\in\{1,2,3\}$. We set $N=n=500$ or $N=n=1000$ separately and generate $1000$ realizations for each setting.

Table \ref{tab:1} reports the absolute average bias (Bias), standard error (SE), and coverage probability (CP) of the 95\% CI of the DR and PAD estimators. When at least one nuisance models are correct, DR and PAD attain very close bias, which is much smaller compared to their SE and, thus, grants their CPs to be close to the nominal level. This indicates that PAD achieves the double robustness property just like the standard DR estimator under finite samples. To compare PAD and DR in terms of their estimation variance and efficiency, we present in Table \ref{tab:2} their relative efficiency (RE) defined as $\Var(\widehat\mu\subDR)/\Var(\widehat\mu\subPAD)$. Under Settings (G1), (G2), (L1), and (L2) where the PS model is correct, the two estimators show nearly identical variance, with their REs located between $1\pm 0.04$. Under Settings (G3) and (L3) with misspecified PS and correct OR models, our proposed PAD estimator shows 20\% to 40\% smaller variance than the standard DR estimator. All these results demonstrate that conclusions in Theorem \ref{thm:1} also apply well for finite samples. In specific, PAD performs very closely to the standard DR when the PS model is correct and is potentially better than DR in the presence of wrong PS models.

	\begin{table}[htb]
		\centering
		\caption{\label{tab:1} The absolute average bias (Bias), standard error (SE), and coverage probability (CP) of the 95\% confidence intervals of the DR and PAD estimators under the settings described in Section \ref{sec:simu}. All results are produced based on $1000$ repetitions.}
		\setlength{\tabcolsep}{3.5mm}{
			\begin{tabular}{@{}cccccccccc@{}}
				\toprule
				&  &                 \multicolumn{3}{c}{$n=N=500$} & \multicolumn{3}{c}{$n=N=1000$} 
				\\ \cmidrule(lr){3-5} \cmidrule(lr){6-8} 
				Setting  &     Method  & Bias   & SE  & CP  & bias   & SE  & CP
				\\ \midrule
                 (G1) & DR & 0.006 & 0.145 & 0.94 & 0.005 & 0.106 & 0.92
                \\
                &PAD & 0.005 & 0.142 & 0.93 &0.004 & 0.105 & 0.92
                \\ \midrule
                 (G2)  & DR & 0.007 & 0.152 & 0.92 & 0.008 & 0.111 & 0.92
                \\
                &PAD & 0.005 & 0.149  & 0.92 & 0.007 & 0.112 & 0.92
                \\ \midrule
                (G3) & DR & 0.010 & {\bf 0.162} &0.93 & 0.001 & {\bf 0.121} &0.92
                \\
                &PAD & 0.005 & {\bf 0.136}  & 0.93 &0.001 &{\bf 0.105} &0.93
                \\ \midrule
                (L1) & DR & 0.000 & 0.055 &0.92 & 0.001 &0.040 &0.92
                \\
                &PAD & 0.001 & 0.054 & 0.93& 0.001&0.040 &0.93
                \\ \midrule
                (L2) &DR &0.001 & 0.054 &0.92 & 0.004 &0.040 &0.92
                \\
                &PAD & 0.001 & 0.053 & 0.92&0.004 &0.040  &0.92
              \\ \midrule
                (L3) & DR & 0.005 & {\bf 0.057} &0.91 & 0.003 &  {\bf 0.038} &0.92
                \\
                &PAD & 0.005 & {\bf 0.052} & 0.93&0.002 &  {\bf 0.035} &0.93
                
				\\ \bottomrule
		\end{tabular}}
	\end{table} 

 \begin{table}[htb!]
		\centering
		\caption{\label{tab:2} Relative efficiency (RE) between DR and PAD, i.e., $\Var(\widehat\mu\subDR)/\Var(\widehat\mu\subPAD)$, under the settings described in Section \ref{sec:simu}.}
		\setlength{\tabcolsep}{4mm}{
			\begin{tabular}{@{}ccccccc@{}}
				\toprule
				$n,N$ & (G1) & (G2) & (G3) & (L1) & (L2) & (L3)
                \\ \midrule
                $500$ & 1.04 &1.04 & {\bf 1.42} & 1.04 & 1.04 & {\bf 1.20}
                 \\ 
                $1000$ & 1.02  &0.98  & {\bf 1.33} & 1.00 & 1.00 & {\bf 1.18}
				\\ \bottomrule
		\end{tabular}}
	\end{table}

\section{Real example}
The effects of the 401(k) program have been investigated for a long time \citep[e.g.]{abadie2003semiparametric,chernozhukov2018double}. Different from other plans like Individual Retirement Accounts (IRAs), eligibility for 401(k) is completely decided by employers. Therefore, unobserved personal preferences for savings may make little difference in 401(k) eligibility. However, there may be some other confounders affecting the causal studies of 401(k), such as job choice, income, and age. To address this problem, \citep{abadie2003semiparametric} and \citep{chernozhukov2018double} proposed to adjust for certain covariates related to job choice so that 401(k) eligibility can be regarded exogenous.

Whether 401(k) eligibility contributes to the improvement of people's net total financial assets is an important topic studied in existing literature like \cite{abadie2003semiparametric} and \cite{chernozhukov2018double}. However, whether 401(k) can improve the financial assets of those actually not eligible for 401(k) is still an open and interesting problem. To investigate this problem, we analyze the data from the Survey of Income and Program Participation of 1991. The data set consists of $n+N=9275$ observations. The outcome of our interests, $Y$ is defined as the indication of having positive net total financial assets. There are $9$ adjustment covariates in $\X$, including age, income, family size, years of education, benefit pension status, marriage, two-earner household status, individual participation in IRA plan, and home ownership status. The source (treated) samples $\Ssc$ with $\Delta=1$ are taken as those eligible for 401(k) and the target (untreated) samples $\Tsc$ are those without 401(k) eligibility. We applied PAD and standard DR to estimate $\mu$, the effect of 401(k) eligibility on improving the positive rate of net total financial assets among people without 401(k) eligibility. The PS model is specified as $\exp(\X\trans\bgamma)$ and the OR model is ${\rm expit}(\X\trans\balpha)$. In our method, the augmentation covariates vector $\bPhi$ consists of $\X$, $\exp(-0.3X_j)$, $|X_j|$, and $X_j^2$ for all $X_j$'s that are not binary. We again use bootstrap to estimate SEs and construct CIs.

In Table \ref{tab:3}, we report the point estimation, their estimated standard errors (ESE), and 95\% CIs for the treatment effect $\mu$, obtained using the standard DR and our proposed PAD methods. Outputs of both methods indicate that 401(k) eligibility has a significant effect on improving the rate of having positive net total financial assets among people who are actually not eligible for 401(k). The estimated treatment effect is $0.169$ (95\% CI: $0.142, 0.196$) by the standard DR and $0.150$ (95\% CI: $0.126, 0.175$) by PAD. Moreover, the ESE of our proposed PAD estimator is remarkably smaller than that of the standard DR estimator, with their estimated RE, i.e., $\Var(\widehat\mu\subDR)/\Var(\widehat\mu\subPAD)$ being around $1.25$. This means our proposed PAD method can characterize the treatment effect $\mu$ more precisely than DR in this example. 

\begin{table}[h]
		\centering
		\caption{\label{tab:3} The point estimation (PE), its estimated standard error (ESE), and 95\% confidence interval (CI) for $\mu$, the effect of 401(k) eligibility on improving the positive rate of net total financial assets among people without 401(k) eligibility, derived using the standard DR and the PAD methods.}
		\setlength{\tabcolsep}{6mm}{
			\begin{tabular}{@{}cccc@{}}
				\toprule
				Method & PE & ESE & CI
                \\ \midrule
                DR& $0.169$  &  $0.0140$  & $(0.142, 0.196)$
                \\
                PAD& $0.150$  &  $0.0125$  & $(0.126, 0.175)$
				\\ \bottomrule
		\end{tabular}}
	\end{table} 

\section{Discussion}

In analogy to our PS model augmentation strategy, we also propose an OR model augmentation strategy (OAD) that augments the OR model with some bases of $\X$ satisfying certain moment conditions like $\bPsi$ in Algorithm \ref{alg:1}. Description and discussion of this method are presented in Section \ref{sec:app:dual} of the Appendix. Similar to Theorem \ref{thm:1}, we are able to show that this OAD estimator is doubly robust, of a smaller variance than the standard DR estimator when the PS model is correct but the OR model is wrong, and equivalent with DR when both nuisance models are correct. Just like PAD, this OAD method is easy to implement and only requires convex optimization. We notice that some existing methods in intrinsic efficient DR estimation like \cite{rotnitzky2012improved} and \cite{gronsbell2022efficient} rely on non-convex training to construct the OR model when it is not linear. This OAD strategy could mitigate this practical problem and still achieves the purpose of variance reduction in the presence of misspecified OR models. 

For ease of demonstration, we focus on covariate shift correction, or equivalently ATT estimation in this paper. Our proposed PAD estimation can be potentially generalized to address other causal or missing data problems like ATE estimation \citep[e.g.]{bang2005doubly}, casual model estimation \cite{rotnitzky2012improved}, transfer learning of a regression model \cite{liu2020doubly}, etc. Also, properly specifying the bases $\bPhi$ is crucial for variance reduction in our method. The optimal choice of $\bPhi$ for the most effective variance reduction is still an open problem. Related to this, it may be useful and interesting to extend our current framework for high-dimensional sparse or sieve construction of the augmentation term $\bPsi\trans\bbeta$.

\bibliographystyle{apalike}
\bibliography{library}

% Appendix

\clearpage
\newpage
\setcounter{page}{1}
\appendix

\setcounter{lemma}{0}
\setcounter{equation}{0}
\setcounter{theorem}{0}
\setcounter{figure}{0}
\setcounter{algorithm}{0}
\setcounter{table}{0}
\renewcommand{\thefigure}{A\arabic{figure}}
\renewcommand{\thetable}{A\arabic{table}}
\renewcommand{\theequation}{A\arabic{equation}}
\renewcommand{\thelemma}{A\arabic{lemma}}
\renewcommand{\thetheorem}{A\arabic{theorem}}
\renewcommand{\thealgorithm}{A\arabic{algorithm}}

\setcounter{definition}{0}
\renewcommand{\thedefinition}{A\arabic{definition}}

%\section*{Appendix}

\section{Dual construction to augment OR}\label{sec:app:dual}
In analogy to our PAD estimator, to improve the efficiency our the DR estimator under the correct PS and wrong OR models, we propose the Outcome regression Augmented Doubly robust (OAD) estimator in the following algorithm. 
\begin{algorithm}[htb!]
\caption{\label{alg:3} Outcome regression Augmented Doubly robust (OAD) estimation}
\begin{algorithmic}
 \State[Step 1] Solve the estimating equations in (\ref{equ:2.1.1}) to obtain $\widehat\bgamma$ and $\widehat\balpha$, and obtain the conditional variance estimator as $\widehat\btheta$.
~\\
 \State[Step 2] Let $\bPhi=\bphi(\X)$ with function $\bphi(\cdot)$, $\widetilde g(\X\trans\widehat\alpha)=g(\X\trans\widehat\alpha)-\widehat\EE_{\Tsc}g(\X\trans\widehat\alpha)$ and
 $$
 \widehat\bPsi =\bPhi-\frac{\widehat\EE_{\Tsc}\bPhi \widetilde g(\X\trans\widehat\alpha)}{\widehat\EE_{\Tsc}\widetilde g^2(\X\trans\widehat\alpha)}\widetilde g(\X\trans\widehat\alpha).
 $$
 \State[Step 3] Solve the restricted weighted least square (RWLS) problem:
 \begin{equation}
 \widehat\bbeta={\rm argmin}_{\bbeta} \widehat V_{\mu,\OAD}(\bbeta),\quad\mbox{s.t.}\quad\widehat\EE_{\Ssc}\X\widehat\bPsi\trans\bbeta\exp(\X\trans\widehat\bgamma)=\bzero,
 \label{equ:rwls:oad}
 \end{equation}
where
\begin{equation}
\begin{aligned}
    \widehat V_{\mu,\OAD}(\bbeta)=&n^{-1}\widehat\Var_{\Ssc}[\{Y-g(\X\trans\widehat\balpha)-\widehat\bPsi\trans\widehat\bbeta\}\exp(\X\trans\widehat\bgamma)]+N^{-1}\widehat\Var_{\Tsc}\{g(\X\trans\widehat\balpha)+\widehat\bPsi\trans\widehat\bbeta\}
    \\&+2\widehat{\bL^*}\trans[N^{-1} \widehat \Cov_{\Tsc}(\X,\widehat\bPsi\trans\widehat\bbeta)+n^{-1}\widehat\Cov_{\Ssc}\{\X\exp(\X\trans\widehat\bgamma),\widehat \bPsi\trans\widehat\bbeta\exp(\X\trans\widehat\bgamma)\}],
\end{aligned}
\label{equ:vbeta:oad}
\end{equation}
and $\widehat{\bL^*}=\{\widehat\EE_{\Ssc}\X\exp(\X\trans\widehat\bgamma)\X\trans\}^{-1} \widehat\EE_{\Ssc}\{Y-g(\X\trans\widehat\balpha)\}\exp(\X\trans\widehat\bgamma)\X$.
 ~\\
 \State[Step 4] Obtain the OAD estimator:
\[
 \widehat\mu\subOAD=\widehat\EE_{\Ssc}\{Y-g(\X\trans\widehat\balpha)-\widehat\bPsi\trans\widehat\bbeta\}\exp(\X\trans\widehat\bgamma)
   +\widehat\EE_{\Tsc}\{g(\X\trans\widehat\balpha)+\widehat\bPsi\trans\widehat\bbeta\}.
 \] 
\end{algorithmic}
\end{algorithm}

To demonstrate how Algorithm \ref{alg:3} works, we define that
 \begin{align*}
\widetilde\mu\subOAD=&\widehat\EE_{\Ssc}\{Y-g(\X\trans\bar\balpha)-\bPsi\trans\bar\bbeta\}\exp(\X\trans\bar\bgamma)
   +\widehat\EE_{\Tsc}\{g(\X\trans\bar\balpha)+\bPsi\trans\bar\bbeta\}
\\&+\EE_{\Ssc}\{Y-g(\X\trans\bar\balpha)\}\exp(\X\trans\bar\bgamma)\X\trans\{\EE_{\Ssc}\X\exp(\X\trans\bar\bgamma)\X\trans\}^{-1}\{\widehat\EE_{\Tsc}\X-\widehat\EE_{\Ssc}\X\exp(\X\trans\bar\bgamma)\}.
\end{align*}
Then similar to our analysis in Section \ref{sec:method:expand}, when the PS model is correct, $\widehat\mu\subOAD$ is asymptotically equivalent to $\widetilde\mu\subOAD$, and
\begin{align*}
    V_{\mu,\OAD}(\bbeta)=&n^{-1}\Var_{\Ssc}[\{Y-g(\X\trans\bar\balpha)-\bPsi\trans\bar\bbeta\}\exp(\X\trans\bar\bgamma)]+N^{-1}\Var_{\Tsc}\{g(\X\trans\bar\balpha)+\bPsi\trans\bar\bbeta\}
    \\&+2{\bL^*}\trans[N^{-1}\Cov_{\Tsc}(\X,\bPsi\trans\bar\bbeta)+n^{-1}\Cov_{\Ssc}\{\X\exp(\X\trans\bar\bgamma),\bPsi\trans\bar\bbeta\exp(\X\trans\bar\bgamma)\}],
\end{align*}
the limiting function of $\widehat V_{\mu,\OAD}(\bbeta)$ specified in Algorithm \ref{alg:3}, where 
\[
\bL^*=\{\EE_{\Ssc}\X\exp(\X\trans\bar\bgamma)\X\trans\}^{-1} \EE_{\Ssc}\{Y-g(\X\trans\bar\balpha)\}\exp(\X\trans\bar\bgamma)\X.
\]
This corresponds to the objective function in equation (\ref{equ:vbeta:oad}). Similar to the PAD construction, when $\bbeta=\bzero$, $V_{\mu,\OAD}(\bbeta)$ reduces to the asymptotic variance of the standard DR estimator (with a constant difference invariant with $\bbeta$). Thus, $\widehat\mu\subOAD$ has a smaller variance than the standard DR estimator when the PS model is correct and the OR model is wrong, under which we typically have $\bar\bbeta\neq \bzero$.

On the other hand, when OR is correctly specified, we have $\bar\balpha=\balpha_0$, $\bL^*=\bzero$, and thus
\[
\frac{\partial V_{\mu,\OAD}(\bbeta)}{\partial\bbeta}|_{\bbeta=\bzero}=\Cov_{\Tsc}(g(\X\trans\bar\balpha),\bPsi).
\]
By definition of $\bPsi$, we have $\Cov_{\Tsc}(g(\X\trans\bar\balpha),\bPsi)=\bzero$. Hence, similar to the analysis in Section \ref{sec:method:expand}, $\widehat\mu\subOAD$ preserved the same DR property as $\widehat\mu\subDR$, i.e., being root-$n$ consistent whenever the PS or the OR model is correctly specified.

\setcounter{lemma}{0}
\setcounter{equation}{0}
\setcounter{theorem}{0}
\setcounter{figure}{0}
\setcounter{table}{0}
\renewcommand{\thefigure}{B\arabic{figure}}
\renewcommand{\thetable}{B\arabic{table}}
\renewcommand{\theequation}{B\arabic{equation}}
\renewcommand{\thelemma}{B\arabic{lemma}}
\renewcommand{\thetheorem}{B\arabic{theorem}}

\setcounter{definition}{0}
\renewcommand{\thedefinition}{B\arabic{definition}}

\section{Asymptotic justification}\label{sec:proof}
\subsection{Technical lemma}

\begin{lemma}
Define $\ba:=\EE_{\Ssc}\bPsi\X\trans \dot g(\X\trans\bar\balpha) $, $\bb:=\EE_{\Ssc}\bPsi \exp(\X\trans\bar\bgamma)v_{\bar \btheta}(\X)+\EE_{\Ssc}\bPsi\X\trans v_{\bar \btheta}(\X)\bL$, and $\bSigma:=\EE_{\Ssc}\bPsi\bPsi\trans v_{\bar \btheta}(\X)$, under Assumption (\ref{asu:3}), the solution of the RWLS problem (\ref{equ:rwls}) is 
\[
\bar\bbeta=\bSigma^{-1}\ba(\ba\trans\bSigma^{-1}\ba)^{-1}\ba\trans\bSigma^{-1}\bb-\bSigma^{-1}\bb.
\]
\label{lem:b1}
\end{lemma}
\begin{proof}
First we introduce Lagrange multiplier $\blambda$ and write (\ref{equ:rwls}) as the Lagrange form:
\[
\bar\bbeta={\rm argmin}_{\bbeta}~\EE_{\Ssc}\{\exp(\X\trans\bar\bgamma)+\bPsi\trans\bbeta\}^2 v_{\bar \btheta}(\X)+2\bL\trans \EE_{\Ssc}\X\{\exp(\X\trans\bar\bgamma)+\bPsi\trans\bbeta\}v_{\bar \btheta}(\X)-\blambda\trans\EE_{\Ssc}\X \dot g(\X\trans\bar\balpha)\bPsi\trans\bbeta.
\]
Then we have the partial derivative of $\blambda$ and $\bbeta$:
\begin{equation}
	\EE_{\Ssc}\X \dot g(\X\trans\bar\balpha)\bPsi\trans\bbeta=\bzero,
\label{equ:8.1.1}
\end{equation}
and
\begin{equation}
2\EE_{\Ssc}\bPsi\{\exp(\X\trans\bar\bgamma)+\bPsi\trans\bbeta\}v_{\bar \btheta}(\X)+2 \EE_{\Ssc}\bPsi\X\trans v_{\bar \btheta}(\X)\bL-\EE_{\Ssc}\bPsi\X\trans \dot g(\X\trans\bar\balpha)\blambda=\bzero.
\label{equ:8.2.1}
\end{equation}
From (\ref{equ:8.2.1}) we have
\[
\bbeta=\{2\EE_{\Ssc}\bPsi\bPsi\trans v_{\bar \btheta}(\X)\}^{-1}\{\EE_{\Ssc}\bPsi\X\trans \dot g(\X\trans\bar\balpha)\blambda-2 \EE_{\Ssc}\bPsi \exp(\X\trans\bar\bgamma)v_{\bar \btheta}(\X)-2 \EE_{\Ssc}\bPsi\X\trans v_{\bar \btheta}(\X)\bL\},
\]
together with (\ref{equ:8.1.1}), we have
\begin{align*}
   \EE_{\Ssc}\X \dot g(\X\trans\bar\balpha)\bPsi\trans&\{\EE_{\Ssc}\bPsi\bPsi\trans v_{\bar \btheta}(\X)\}^{-1}
   \\
   &*\{\EE_{\Ssc}\bPsi\X\trans \dot g(\X\trans\bar\balpha)\blambda-2 \EE_{\Ssc}\bPsi \exp(\X\trans\bar\bgamma)v_{\bar \btheta}(\X)-2 \EE_{\Ssc}\bPsi\X\trans v_{\bar \btheta}(\X)\bL\}=\bzero.
\end{align*}
this function can be simplified as
\[
\ba\trans\bSigma^{-1}(\ba\blambda-2\bb)=\bzero,
\]
and we further have 
\[
\blambda=2(\ba\trans\bSigma^{-1}\ba)^{-1}\ba\trans\bSigma^{-1}\bb.
\]
Hence, we have
\[
\bar\bbeta=\bSigma^{-1}\ba(\ba\trans\bSigma^{-1}\ba)^{-1}\ba\trans\bSigma^{-1}\bb-\bSigma^{-1}\bb.
\]
%and we can estimate it by 
%\[
%\widehat\bbeta=\widehat\bSigma^{-1}\widehat\ba(\widehat\ba\trans\widehat\bSigma^{-1}\widehat\ba)^{-1}\widehat\ba\trans\widehat\bSigma^{-1}\widehat\bb-\widehat\bSigma^{-1}\widehat\bb
%\]
%for $\widehat\ba:=\widehat\EE_{\Ssc}\widehat\bPsi\X\trans \dot g(\X\trans\widehat\balpha) $, $\widehat\bb:=\widehat\EE_{\Ssc}\widehat\bPsi \exp(\X\trans\widehat\bgamma)v_{\widehat\btheta}(\X)+\widehat\EE_{\Ssc}\widehat\bPsi\X\trans v_{\widehat\btheta}(\X)\widehat\bL$, and $\widehat\bSigma:=\widehat\EE_{\Ssc}\widehat\bPsi\widehat\bPsi\trans v_{\widehat\btheta}(\X)$
\end{proof}
\begin{lemma} 
Under Assumptions (\ref{asu:1}) and (\ref{asu:4}), we have that $\widehat\bPsi-\bPsi=O_p(n^{-1/2})$.
\label{lem:b2}
\end{lemma}
\begin{proof}
By definition, we would have that
\[
\widehat\bPsi-\bPsi=\frac{\widehat\EE_{\Tsc}\{\bPhi v_{\widehat\btheta}(\X)\}}{\widehat\EE_{\Tsc}v_{\widehat\btheta}(\X)}-\frac{\EE_{\Tsc}\{\bPhi  v_{\bar \btheta}(\X)\}}{\EE_{\Tsc} v_{\bar \btheta}(\X)}.
\]
Under Assumption (\ref{asu:4}), we have that 
\begin{equation}
    \widehat\EE_{\Tsc}v_{\widehat\btheta}(\X)-\EE_{\Tsc} v_{\bar \btheta}(\X)=\widehat\EE_{\Tsc}v_{\bar\btheta}(\X)+\widehat\EE_{\Tsc}\frac{\partial v_{\btheta}(\X)}{\partial \btheta}|_{\widetilde\btheta}(\widehat\btheta-\bar\btheta)-\EE_{\Tsc} v_{\bar \btheta}(\X)=O_p(n^{-1/2})
    \label{equ:8.1.2}
\end{equation}
for $\widetilde\btheta$ between $\widehat\btheta$ and $\bar\btheta$. By using the same techniques, we have that $\widehat\EE_{\Tsc}\{\bPhi v_{\widehat\btheta}(\X)\}-\EE_{\Tsc}\{\bPhi  v_{\bar \btheta}(\X)\}=O_p(n^{-1/2})$. 
And we have 
\begin{align*}
&\widehat\bPsi-\bPsi=\frac{\widehat\EE_{\Tsc}\{\bPhi v_{\widehat \btheta}(\X)\}\EE_{\Tsc} v_{\bar \btheta}(\X)-\EE_{\Tsc}\{\bPhi  v_{\bar \btheta}(\X)\}\widehat\EE_{\Tsc} v_{\widehat \btheta}(\X)}{\widehat\EE_{\Tsc} v_{\widehat \btheta}(\X)\EE_{\Tsc} v_{\bar \btheta}(\X)}
\\
&=\frac{\widehat\EE_{\Tsc}\{\bPhi  v_{\widehat \btheta}(\X)\}\EE_{\Tsc} v_{\bar \btheta}(\X)-\EE_{\Tsc}\{\bPhi  v_{\bar \btheta}(\X)\}\EE_{\Tsc} v_{\bar \btheta}(\X)-[\EE_{\Tsc}\{\bPhi  v_{\bar \btheta}(\X)\}\widehat\EE_{\Tsc} v_{\widehat \btheta}(\X)-\EE_{\Tsc}\{\bPhi  v_{\bar \btheta}(\X)\}\EE_{\Tsc} v_{\bar \btheta}(\X)]}{\{\EE_{\Tsc} v_{\bar \btheta}(\X)+O_p(n^{-1/2})\}\EE_{\Tsc} v_{\bar \btheta}(\X)}
\\
&=\frac{O_p(n^{-1/2})\EE_{\Tsc} v_{\bar \btheta}(\X)-\EE_{\Tsc}\{\bPhi  v_{\bar \btheta}(\X)\}O_p(n^{-1/2})}{\{\EE_{\Tsc} v_{\bar \btheta}(\X)+O_p(n^{-1/2})\}\EE_{\Tsc} v_{\bar \btheta}(\X)}=O_p(n^{-1/2}).
\end{align*}
\end{proof}
\begin{lemma}
Under Assumptions (\ref{asu:1}) and (\ref{asu:2}), we have that $\widehat\bgamma-\bar\bgamma=O_p(n^{-1/2})$ and $\widehat\balpha-\bar\balpha=O_p(n^{-1/2})$.
\label{lem:b3}
\end{lemma}
\begin{proof}
The estimation of $\bgamma$ has been given as
\[
\widehat\EE_{\Ssc}\X\exp(\X\trans\widehat\bgamma)=\widehat\EE_{\Tsc}\X,
\]
by applying Taylor series expansion, we have
\[
n^{-1}\sum_{i=1}^n\X_i\exp(\X_i\trans\bar\bgamma)+n^{-1}\sum_{i=1}^n\X_i\exp(\X_i\trans\widetilde\bgamma)\X_i\trans(\widehat\bgamma-\bar\bgamma)=N^{-1}\sum_{i=n+1}^{n+N}\X_i,
\]
where $\widetilde\bgamma$ is some vector between $\widehat\bgamma$ and $\bar\bgamma$. According to \citep[Chapter 5]{van2000asymptotic}, we have $\widehat\bgamma-\bar\bgamma=o_p(1)$. Let $\bJ$ represent matrix $n^{-1}\sum_{i=1}^n\X_i\exp(\X_i\trans\widetilde\bgamma)\X_i\trans$, and we have that 
\[
\bJ=n^{-1}\sum_{i=1}^n\X_i\exp(\X_i\trans\bar\bgamma)\X_i\trans+n^{-1}\sum_{i=1}^n\X_i\exp(\X_i\trans\bgamma^*)\X_i\trans\X_i(\widetilde\bgamma-\bar\bgamma)=\EE_{\Ssc}\X\exp(\X\trans\bar\bgamma)\X\trans+o_p(1)
\]
for $\bgamma^*$ between $\widetilde\bgamma$ and $\bar\bgamma$. Hence, by central limit theorem and Slutsky theorem, we have that, 
\begin{align*}
	&\widehat\bgamma-\bar\bgamma=\bJ^{-1}\bigg\{N^{-1}\sum_{i=n+1}^{n+N}\X_i-n^{-1}\sum_{i=1}^n\X_i\exp(\X_i\trans\bar\bgamma)\bigg\}\\
	=&\bJ^{-1}\bigg\{N^{-1}\sum_{i=n+1}^{n+N}\X_i-\EE_{\Tsc}\X+\EE_{\Ssc}\X\exp(\X\trans\bar\bgamma)-n^{-1}\sum_{i=1}^n\X_i\exp(\X_i\trans\bar\bgamma)\bigg\}=O_p(n^{-1/2}).
\end{align*}
Furthermore, 
The estimation equation of $\widehat\alpha$ is given by
\[
\widehat\EE_{\Ssc}\bS(\widehat\balpha)=\widehat\EE_{\Ssc}\X\{Y-g(\X\trans\widehat\balpha)\}=\bzero,
\]
by using Taylor series expansion, we have that \[
\widehat\EE_{\Ssc}\X\{Y-g(\X\trans\bar\balpha)\}+\widehat\EE_{\Ssc}\frac{\partial \bS(\balpha)}{\partial \balpha\trans}\bigg|_{\widetilde\alpha}(\widehat\balpha-\bar\balpha)=\bzero
\]
for $\widetilde\alpha$ between $\widehat\balpha$ and $\bar\balpha$, and we have
\[
\widehat\balpha-\bar\balpha=-\widehat\EE_{\Ssc}\bigg\{\frac{\partial \bS(\balpha)}{\partial \balpha\trans}\bigg|_{\widetilde\balpha}\bigg\}^{-1}\widehat\EE_{\Ssc}\X\{Y-g(\X\trans\bar\balpha)\}.
\]
By using the same techniques as those for obtaining the asymptotic properties of $\widehat\bgamma$, under Assumptions (\ref{asu:1}) and (\ref{asu:2}), we have $\widehat\balpha-\bar\balpha=O_p(n^{-1/2})$.
\end{proof}
\begin{lemma}
Under Assumptions (\ref{asu:1})-(\ref{asu:4}) and Lemma (\ref{lem:b1})-(\ref{lem:b3}), we can obtain that $\widehat\bbeta-\bar\bbeta=O_p(n^{-1/2})$. In addition, when the PS is correctly specified, we further have $\bar\bbeta=\bzero$ and $\widehat\bbeta=O_p(n^{-1/2})$.
\label{lem:b4}
\end{lemma}
\begin{proof}
By using the same techniques as (\ref{equ:8.1.2}), under Condition 2-4, we first have that
\[
\widehat\ba-\ba=\widehat\EE_{\Ssc}\widehat\bPsi\X\trans \dot g(\X\trans\widehat\balpha)-\EE_{\Ssc}\bPsi\X\trans \dot g(\X\trans\bar\balpha)=\widehat\EE_{\Ssc}\bPsi\X\trans \dot g(\X\trans\bar\balpha)-\EE_{\Ssc}\bPsi\X\trans \dot g(\X\trans\bar\balpha)+O_p(n^{-1/2})=O_p(n^{-1/2}).
\]
In addition, we can have that $\widehat\bb-\bb=O_p(n^{-1/2})$ and $\bSigmahat-\bSigma=O_p(n^{-1/2})$. Furthermore, we can easily have that 
\begin{align*}
	\bSigmahat^{-1}-\bSigma^{-1}&=\bSigma^{-1}\bSigma\{\bSigma+O_p(n^{-1/2})\}^{-1}-\bSigma^{-1}
	\\&=\bSigma^{-1}[\bSigma\{\bSigma+O_p(n^{-1/2})\}^{-1}-\{\bSigma+O_p(n^{-1/2})\}\{\bSigma+O_p(n^{-1/2})\}^{-1}]=O_p(n^{-1/2}),
\end{align*}
based on which we can have $(\widehat\ba\trans\bSigmahat^{-1}\widehat\ba)^{-1}-(\ba\trans\bSigma^{-1}\ba)^{-1}=O_p(n^{-1/2})$. Let $\widehat\bOmega$ denote $\bSigmahat^{-1}\widehat\ba(\widehat\ba\trans\bSigmahat^{-1}\widehat\ba)^{-1}\widehat\ba\trans\bSigmahat^{-1}$ and $\bOmega$ denote $\bSigma^{-1}\ba(\ba\trans\bSigma^{-1}\ba)^{-1}\ba\trans\bSigma^{-1}$. We can have that $\widehat\bOmega-\bOmega=O_p(n^{-1/2})$, hence, we have that $\widehat\bbeta-\bar\bbeta=\widehat\bOmega\widehat\bb-\bOmega\bb=O_p(n^{-1/2}).$

On the other hand, when the PS is correctly specified, $\bL=\bzero$ and $\EE_{\Ssc}\bPsi\exp(\X\trans\bar\bgamma)v_{\bar \btheta}(\X)=\EE_{\Tsc}\bPsi v_{\bar \btheta}(\X)=\bzero$, which means 
\[
\bar\bbeta=\bOmega\bb=\bOmega\{\EE_{\Ssc}\bPsi \exp(\X\trans\bar\bgamma)v_{\bar \btheta}(\X)+\EE_{\Ssc}\bPsi\X\trans v_{\bar \btheta}(\X)\bL\}=\bOmega\bzero=\bzero.
\]
And at the same time, we have $\widehat\bbeta=O_p(n^{-1/2})$.
\end{proof}

\subsection{Proof of Theorem \ref{thm:1}}

\begin{proof}
Proof of Theorem 1 (i).

When the OR is correctly specified, $\bar\balpha=\balpha_0$. Consider $\widetilde \mu\subOR$ where
\begin{align*}
   &\widetilde\mu\subOR=\widehat\EE_{\Ssc}\{Y-g(\X\trans\bar\balpha)\}\{\exp(\X\trans\bar\bgamma)+\bPsi\trans\bar\bbeta\}+\widehat\EE_{\Tsc}g(\X\trans\bar\balpha)
   \\
   &+\{\EE_{\Ssc}\X\trans \dot g(\X\trans\bar\balpha)\exp(\X\trans\bar\bgamma)-\EE_{\Tsc}\X\trans \dot g(\X\trans\bar\balpha)\}\EE_{\Ssc}\bigg\{\frac{\partial \bS(\balpha)}{\partial \balpha\trans}\bigg|_{\bar\balpha}\bigg\}^{-1}\widehat\EE_{\Ssc}\X\{Y-g(\X\trans\bar\balpha)\}.
\end{align*}
It is obvious that $\EE\widetilde\mu\subOR=\EE_{\Tsc}g(\X\trans\bar\balpha)=\mu_0$. Hence, by using central limit theorem, we have that $\widetilde\mu\subOR-\mu_0=O_p(n^{-1/2})$, $n^{1/2}(\widetilde\mu\subOR-\mu_0)$ weakly converges to gaussian distribution with mean $\bzero$. On the other hand, we have that 
 \begin{align*}
     &\widehat\mu\subPAD-\widetilde\mu\subOR=\widehat\EE_{\Ssc}\{Y-g(\X\trans\balpha_0)\}\{\exp(\X\trans\bar\bgamma)\X\trans(\widehat\bgamma-\bar\bgamma)+\bPsi\trans(\widehat\bbeta-\bar\bbeta)+(\widehat\bPsi-\bPsi)\trans\bar\bbeta\}
     \\
     &-[\widehat\EE_{\Ssc}\X\trans \dot g(\X\trans\balpha_0)\{\exp(\X\trans\bar\bgamma)+\bPsi\trans\bar\bbeta\}-\widehat\EE_{\Tsc}\X\trans \dot g(\X\trans\balpha_0)](\widehat\balpha-\bar\balpha)+o_p(n^{-1/2})
     \\
     &-\{\EE_{\Ssc}\X\trans \dot g(\X\trans\balpha_0)\exp(\X\trans\bar\bgamma)-\EE_{\Tsc}\X\trans \dot g(\X\trans\balpha_0)\}\EE_{\Ssc}\bigg\{\frac{\partial \bS(\balpha)}{\partial \balpha\trans}\bigg|_{\balpha_0}\bigg\}^{-1}\widehat\EE_{\Ssc}\X\{Y-g(\X\trans\balpha_0)\},
 \end{align*}
 by using central limit theorem, along with Lemma (\ref{lem:b1})-(\ref{lem:b4}), we have that 
\begin{align*}
    &\widehat\EE_{\Ssc}\{Y-g(\X\trans\balpha_0)\}\{\exp(\X\trans\bar\bgamma)\X\trans(\widehat\bgamma-\bar\bgamma)+\bPsi\trans(\widehat\bbeta-\bar\bbeta)+(\widehat\bPsi-\bPsi)\trans\bar\bbeta\}
    \\
    &=[\widehat\EE_{\Ssc}\{Y-g(\X\trans\balpha_0)\}\exp(\X\trans\bar\bgamma)\X\trans](\widehat\bgamma-\bar\bgamma)+[\widehat\EE_{\Ssc}\{Y-g(\X\trans\balpha_0)\}\bPsi\trans](\widehat\bbeta-\bar\bbeta)
    \\&+[\widehat\EE_{\Ssc}\{Y-g(\X\trans\balpha_0)\}\bar\bbeta\trans](\widehat\bPsi-\bPsi)=O_p(n^{-1/2})o_p(1)+O_p(n^{-1/2})o_p(1)+O_p(n^{-1/2})o_p(1)=o_p(n^{-1/2}).
\end{align*}
 On the other hand,
 \begin{equation}
\begin{aligned}
     &- [\widehat\EE_{\Ssc}\X\trans \dot g(\X\trans\balpha_0)\{\exp(\X\trans\bar\bgamma)+\bPsi\trans\bar\bbeta\}-\widehat\EE_{\Tsc}\X\trans \dot g(\X\trans\balpha_0)](\widehat\balpha-\bar\balpha)\\
    &=[\widehat\EE_{\Ssc}\X\trans \dot g(\X\trans\balpha_0)\{\exp(\X\trans\bar\bgamma)+\bPsi\trans\bar\bbeta\}-\widehat\EE_{\Tsc}\X\trans \dot g(\X\trans\balpha_0)]\widehat\EE_{\Ssc}\bigg\{\frac{\partial \bS(\balpha)}{\partial \balpha\trans}\bigg|_{\bar\balpha}\bigg\}^{-1}\widehat\EE_{\Ssc}\X\{Y-g(\X\trans\bar\balpha)\}\\
    &=\{\EE_{\Ssc}\X\trans \dot g(\X\trans\balpha_0)\exp(\X\trans\bar\bgamma)-\EE_{\Tsc}\X\trans \dot g(\X\trans\balpha_0)+O_p(n^{-1/2})\}
    \\
    &~~~~~~~~~~~~~~~~~~~~~~~~~~~~~~~~~~~~~~~~~~~~~~~*\bigg[\EE_{\Ssc}\bigg\{\frac{\partial \bS(\balpha)}{\partial \balpha\trans}\bigg|_{\balpha_0}\bigg\}^{-1}+O_p(n^{-1/2})\bigg]\widehat\EE_{\Ssc}\X\{Y-g(\X\trans\balpha_0)\}.
    \label{equ:8.1.3}
    \end{aligned}
 \end{equation}
 Hence, we have that
 \begin{align*}
     &-[\widehat\EE_{\Ssc}\X\trans \dot g(\X\trans\balpha_0)\{\exp(\X\trans\bar\bgamma)+\bPsi\trans\bar\bbeta\}-\widehat\EE_{\Tsc}\X\trans \dot g(\X\trans\balpha_0)](\widehat\balpha-\bar\balpha)
     \\
     &-\{\EE_{\Ssc}\X\trans \dot g(\X\trans\balpha_0)\exp(\X\trans\bar\bgamma)-\EE_{\Tsc}\X\trans \dot g(\X\trans\balpha_0)\}\EE_{\Ssc}\bigg\{\frac{\partial \bS(\balpha)}{\partial \balpha\trans}\bigg|_{\balpha_0}\bigg\}^{-1}\widehat\EE_{\Ssc}\X\{Y-g(\X\trans\balpha_0)\}
     \\
     &=\widehat\EE_{\Ssc}\X\{Y-g(\X\trans\balpha_0)\}O_p(n^{-1/2})=o_p(n^{-1/2}).
 \end{align*}
 Thus, from previous results, we have that $\widehat\mu\subPAD-\widetilde\mu\subOR=o_p(n^{-1/2})$. Together with Slutsky theorem, we futher have that $\widehat\mu\subPAD-\mu_0=O_p(n^{-1/2})$ and $n^{1/2}(\muhat\subPAD-\mu_0)$ weakly converges to gaussian distribution with mean $\bzero$.
 
 When the PS is correctly specified, $\bar\bgamma=\bgamma_0$, we consider $\widetilde\mu\subPS$ where 
 \begin{align*}
   &\widetilde\mu\subPS=\widehat\EE_{\Ssc}\{Y-g(\X\trans\bar\balpha)\}\{\exp(\X\trans\bgamma_0)+\bPsi\trans\bar\bbeta\}
   +\widehat\EE_{\Tsc}g(\X\trans\bar\balpha)
\\&+\EE_{\Ssc}\{Y-g(\X\trans\bar\balpha)\}\exp(\X\trans\bar\bgamma)\X\trans\{\EE_{\Ssc}\X\exp(\X\trans\bgamma_0)\X\trans\}^{-1}\{\widehat\EE_{\Tsc}\X-\widehat\EE_{\Ssc}\X\exp(\X\trans\bgamma_0)\}.
\end{align*}
Together with the results from Lemma (\ref{lem:b4}), we have that $\EE\widetilde\mu\subPS=\EE_{\Ssc}Y\exp(\X\trans\bar\bgamma)=\EE_{\Tsc}Y=\mu_0$. By using the central limit theorem, we have that $\widetilde\mu\subPS-\mu_0=O_p(n^{-1/2})$, $n^{1/2}(\widetilde\mu\subPS-\mu_0)$ weakly converges to gaussian distribution with mean $\bzero$. On the other hand, we have that
\begin{align*}
     &\widehat\mu\subPAD-\widetilde\mu\subPS=\widehat\EE_{\Ssc}\{Y-g(\X\trans\bar\balpha)\}\{\exp(\X\trans\bar\bgamma)\X\trans(\widehat\bgamma-\bar\bgamma)+\bPsi\trans(\widehat\bbeta-\bar\bbeta)\}
  \\&
-\EE_{\Ssc}\{Y-g(\X\trans\bar\balpha)\}\exp(\X\trans\bar\bgamma)\X\trans\{\EE_{\Ssc}\X\exp(\X\trans\bgamma_0)\X\trans\}^{-1}\{\widehat\EE_{\Tsc}\X-\widehat\EE_{\Ssc}\X\exp(\X\trans\bgamma_0)\}+o_p(n^{-1/2})
\end{align*}
By using the techniques from (\ref{equ:8.1.3}), we would have 
\begin{align*}
    &\widehat\EE_{\Ssc}\{Y-g(\X\trans\bar\balpha)\}\exp(\X\trans\bar\bgamma)\X\trans(\widehat\bgamma-\bar\bgamma)
\\&-\EE_{\Ssc}\{Y-g(\X\trans\bar\balpha)\}\exp(\X\trans\bar\bgamma)\X\trans\{\EE_{\Ssc}\X\exp(\X\trans\bgamma_0)\X\trans\}^{-1}\{\widehat\EE_{\Tsc}\X-\widehat\EE_{\Ssc}\X\exp(\X\trans\bgamma_0)\}=o_p(n^{-1/2})
\end{align*}
And from Lemma A4, we have $\widehat\bbeta=O_p(n^{-1/2})$. Thus, we have $\widehat\mu\subPAD-\mu_0=O_p(n^{-1/2})$. On the other hand, it is worth noticing that $\widehat\bbeta$ is the continuous function of $\widehat\btheta$, $\widehat\bgamma$ and $\widehat\balpha$, so under central limit theorem and Slutsky theorem, we would have the asymptotic normality of $\widehat\bbeta$. Hence, we further have that $n^{1/2}(\widehat\mu\subPS-\mu_0)$ weakly converges to gaussian distribution with mean $\bzero$. 
\end{proof}
\begin{proof}
Proof of Theorem 1 (ii).

First we denote $\bU$ as
\[
\bU=\Var_{\Tsc}(\EE(Y|\X))+\bL\trans\EE_{\Ssc}\X\X\trans\Var(Y|\X)\bL
\]
When the OR is correctly specified, the asymptotic variance of $\widehat\mu\subPAD$, $\Var\{n^{-1/2}(\widehat\mu\subPAD-\mu_0)\}$ is 
\[
\EE_{\Ssc}\{\exp(\X\trans\bar\bgamma)+\bPsi\trans\bar\bbeta\}^2 v_{\bar\btheta}(\X)+2\bL\trans \EE_{\Ssc}\X\{\exp(\X\trans\bar\bgamma)+\bPsi\trans\bar\bbeta\}v_{\bar\btheta}(\X)+\bU,
\]
and $\bar\bbeta$ contributes to minimizing this variance. When $\bar\bbeta=\bzero$, the function above is written as
\[
\EE_{\Ssc}\{\exp(\X\trans\bar\bgamma)\}^2 v_{\bar\btheta}(\X)+2\bL\trans \EE_{\Ssc}\X\{\exp(\X\trans\bar\bgamma)\}v_{\bar\btheta}(\X)+\bU,
\]
which is the same as the asymptotic variance of $\widehat\mu\subDR$, $\Var\{n^{-1/2}(\widehat\mu\subDR-\mu_0)\}$. Hence, when $\bar\bbeta\neq\bzero$, $\widehat\mu\subPAD$ has the smaller asymptotic variance than standard doubly robust estimator $\widehat\mu\subDR$.
\end{proof}

\begin{proof}
Proof of Theorem 1 (iii).

When both the PS and OR is correctly specified, consider $\widetilde\mu_{B}$, where
\[
\widetilde\mu_{B}=\widehat\EE_{\Ssc}\{Y-g(\X\trans\balpha_0)\}\exp(\X\trans\bgamma_0)+\widehat\EE_{\Tsc}g(\X\trans\balpha_0).
\]
By using central limit theorem, $n^{1/2}(\widetilde\mu_{B}-\mu_0)$ weakly converges to gaussian distribution with mean $\bzero$. On the other hand, by using Taylor series expansion, we would have $\widehat\mu\subPAD-\widetilde\mu_{B}=o_p(n^{-1/2})$ and $\widehat\mu\subDR-\widetilde\mu_{B}=o_p(n^{-1/2})$. Hence, they have the same asymptotic variance. 
\end{proof}

\end{document}